\documentclass[conference]{IEEEtran}
\usepackage {graphicx} % For � kunne inkludere grafikk
\usepackage {subfigure}
\usepackage {amsthm}
\usepackage[fleqn] {amsmath}
\usepackage {amsfonts, amssymb}
\usepackage{multirow}

\usepackage{enumerate}
\usepackage{color}
\usepackage{algorithmic} 
\usepackage[]{algorithm}
\newtheorem{theorem}{Theorem}

\title{The Domino Effect in Decentralized Wireless Networks}
\author{\IEEEauthorblockN{Brage Ellings\ae ter\IEEEauthorrefmark{1} and Torleiv Maseng\IEEEauthorrefmark{2}}
\IEEEauthorblockA{\IEEEauthorrefmark{1}UNIK-University Graduate Center at Kjeller, University of Oslo, Norway}
\IEEEauthorblockA{\IEEEauthorrefmark{2}Norwegian Defence Research Establishment}\texttt{brage@unik.no}}%Email: \href{mailto:brage@unik.no}{brage@unik.no}}
\begin{document}
\maketitle
\begin{abstract}
Convergence of resource allocation algorithms is well covered in the literature as convergence to a steady state is important due to stability and performance. However, research is lacking when it comes to the propagation of change that occur in a network due to new nodes arriving or old nodes leaving or updating their allocation. As change can propagate through the network in a manner similar to how domino pieces falls, we call this propagation of change the domino effect. In this paper we investigate how change at one node can affect other nodes for a simple power control algorithm. We provide analytical results from a deterministic network as well as a Poisson distributed network through percolation theory and provide simulation results that highlight some aspects of the domino effect. The difficulty of mitigating this domino effect lies in the fact that to avoid it, one needs to have a margin of tolerance for changes in the network. However, a high margin leads to poor system performance in a steady-state and therefore one has to consider a trade-off between performance and propagation of change.
\end{abstract}

\section{Introduction}
Convergence to a stable state is desirable when considering resource allocation in decentralized wireless networks since rapid transmitter configuration is a costly operation in terms of overhead and energy consumption. Stable states (or equilibria) also allow us to analyze the performance of specific resource allocation algorithms efficiently in terms of the stable states reached by the algorithms. Thus much research has been devoted to the analysis of convergence of resource allocation algorithms \cite{Foschini1993}\cite{Huang:1998:RCM:274740.274742}\cite{1459062}\cite{Shum2007}\cite{Johansson2012}. The desirable feature of such algorithms is fast convergence rate, convergence to an efficient stable state and convergence independent of the initial conditions. 

An aspect overlooked in the literature on physical layer resource allocation algorithms is the consequence of a change in the network and how it affects the allocation of other nodes in the network. Imagine a new node enters the network or an old node leaves or updates its allocation (e.g. changes power to satisfy a new SNR requirement), how does this affect the nodes in the network? Specifically we are interested in the propagation of change in networks due to such updates. When an update occurs, change propagates through the network in a way similar to a "domino" effect, changing the allocation of nodes in the network.

It is desirable that resource allocation algorithms produce a minimal domino effect, meaning that as few nodes as possible must adapt to a change in the network. Considering large scale networks on the order of a hundred links, if an update at one link results in all other links changing their allocation as a consequence, stable states in the network will not occur if convergence time is longer than the frequency at which updates occur. One way to mitigate such a scenario would be for all nodes to allocate resources with a buffer or margin, similar to a fading margin. However, adding such a margin will lead to a loss in performance in stable states, and it is also difficult to say how large this margin should be, as it depends on a number of variables. Thus, we postulate that there will have to be a trade-off between the performance of steady states and how large a domino effect one can tolerate. To the best of the authors' knowledge, no work has been published in the literature which investigates these issues.  

We start the investigation by presenting a deterministic network where we can analytically characterize the domino effect. We then investigate a network where nodes are distributed in the plane according to a Poisson Point Process. One of our main results are obtained through percolation theory, which has been applied to wireless communication to study the connectivity of wireless multi-hop networks \cite{Dousse2005}\cite{Dousse2006}. We show that there exists a critical density of nodes in the network so that for densities lower than this critical value, the propagation of change is finite almost surely (a.s.), and for densities larger than this critical value, the propagation of change is infinite a.s. Through simulation results we also note that the rounds in which nodes are affected follow a power law distribution for low margins. We can state that for these low margins the system is in a critical state so that the dynamics of the network is similar to systems exhibiting self-organized criticality (SOC) in physics \cite{Bak1988}.

\subsection{Cellular vs. Decentralized}
%\label{sec:cell-vs-dec}
It is important to understand why the domino effect is an important aspect when it comes to designing resource allocation algorithms for decentralized networks. We do this with a comparison to cellular networks, which do not have the problem of a domino effect. In a cellular network the resources are fixed and controlled strictly by the base station controller. E.g. in a CDMA network, the coding gain determines the number of users that can be supported by a base station. When this limit has been reached, a new user is denied access to the network and ripple effects of change is contained within the cell of the base station.

In a wireless network without centralized control, there is no network entity to deny users access to the network resources. A simple example is Wi-Fi networks. Although too many Wi-Fi access points (APs) might be located in a given region for optimal operation, a person can set up a new Wi-Fi access point without any of the other users having a right to complain. For optimal operation after this new WIFI AP has been set up, the surrounding APs might have to change their transmit parameters such as power and frequency.%To increase the performance of WIFI networks and future femtocells dynamic resource allocation is proposed to increase the efficiency of these networks. This means adjusting to the changes in the environment, such as interference power and other sources of noise. %This is the main argument for why such networks exhibits SOC. %{\color{red}add more}

\subsection{An Example of a Self-organized Critical System}
A simple example of a dynamic system exhibiting SOC is the "sand pile" model \cite{Bak1988}. Assume we have $N$ variables, where each variable $z_n$ is set to zero initially. Each variable is thought of as a discrete height plateau in a sand pile. A grain of sand is added to one variable at random. When $z_n$ is greater than a critical value $z_c$, the grain tumbles to the next variable $z_{n+1}$. This is illustrated in Fig. \ref{fig:sand-pile} where $z_c = 3$. In this example, if one grain is added to either $z_1$ or $z_2$ the grain tumbles down to $z_3$, sliding either 1 or 2 plateaus down. In this system the length of a tumble, i.e. how many variables that surpass the critical value, also known as the avalanche size, behaves according to a power law. 

\begin{figure}[t]
\centering
\includegraphics[width = 0.6\columnwidth]{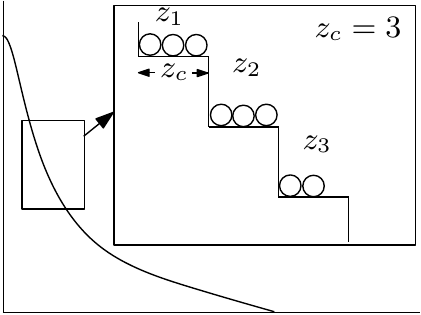}
\caption{Illustration of the sand pile model. Each height plateau is a variable, and each variable can support up to $z_c$ sand grains, in this case $3$.}
\label{fig:sand-pile}
\end{figure}

We can exchange the sand pile model with a set of wireless links, where we assume each link transmits with a power sufficient to support a given SINR. Then assume one link changes its power to support a new SINR requirement. If this change is large enough, "a grain tumbles" meaning some links update their power due to the change. These links are affected in the first "wave". If the accumulated change in power at these links exceeds a given threshold, a new set of links are affected and the avalanche continues to roll.

%The effect this change has on the allocation of the other links resembles the way grain tumbles down sandpile, sometimes it has little effect and sometimes the avalanche is large (affects many links).

\section{System Model and Notation}
\label{sec:model}
%We start by defining a network where all users transmits on the same frequency. Assuming power is a continuous variable within a finite interval, this simplifies the problem as we only have one decision variable and no discrete variables. %A system where frequency is also a decision variable is introduced in Section \ref{sec:soc}. 
We consider a single frequency network where the decision variable at all users is power.
We assume each user consists of a transmitter and a receiver, such that transmitter $i$ transmits to receiver $i$. We denote the set of users as $\mathcal{N}$, where $N = |\mathcal{N}|$. We assume each user $i$ has a SINR requirement $\beta_i$, and transmits with a minimum power $p_i$ such that the SINR requirement is satisfied:
\begin{equation}
SINR_i = \frac{l(d_{ii})p_i}{N_0+\sum_{j\in \mathcal{N},j\neq i} l(d_{ji})p_j}\geq \beta_i
\label{eq:SINR}
\end{equation}
where $d_{ij}$ is the distance between transmitter $i$ and receiver $j$, $l(d_{ij})$ is the channel gain between $i$ and $j$ and $N_0$ is the ambient noise variance.

We use this allocation algorithm due to the fact that when a feasible allocation exist, this minimal power allocation will converge \cite{Foschini1993}. This way we are able to separate the domino effect from the cases when ripples of change occur due to the network not converging.

Let $p_i$ be the transmit power of user $i$ after the network has converged. Next, we assume a random user $x$ updates its transmit power such that $p_x' = p_x+ \Delta$. Let $p_i'$ be the transmit power of user $i$ after the network has converged again after user $x$'s update. We are interested in characterizing the set 
\begin{equation}
\mathcal{A} = \{i|p_i'\neq p_i\}.
\end{equation}
Especially, we are interested in the cardinality of $\mathcal{A}$. We can divide $\mathcal{A}$ into non-overlapping subsets $\mathcal{A}_1$, $\mathcal{A}_2$..., where the users belonging to $\mathcal{A}_1$ are the users directly affected by user $x$'s update, i.e. are affected in the first round. The users belonging to $\mathcal{A}_2$ are the users not directly affected by user $x$'s update, but due to the accumulated change in transmit power in the network from user $x$ and the users in $\mathcal{A}_1$, changes its allocation, i.e. are affected in the second round. %As power is attenuated by at least distance squared, $|\mathcal{A}_i|>>|\mathcal{A}_{i+1}|$ when we only consider power allocation.

Characterizing $\mathcal{A}$ is difficult for multiple reasons. When simulating a network, multiple variables may be modeled as being drawn from different probability distributions. For instance, for a user $i$ to belong to $\mathcal{A}_1$ we have that
\begin{equation}
p_i^1-p_i>\delta_i
\end{equation}
where $\delta_i$ is the threshold for user $i$ to change its power. After user $x$ updates its power, user $i$ must change its power as follows to satisfy its SINR requirement
\begin{equation}
p_i^1 = \frac{\beta_i(N_0+I_i+\Delta l(d_{xi}))}{l(d_{ii})} = p_i + \beta_i\Delta \frac{l(d_{xi})}{l(d_{ii})}
\label{eq:belong-a1}
\end{equation}
where $I_i$ is the old interference at user $i$ as given in the denominator in (\ref{eq:SINR}). And thus user $i$ changes its power if
\begin{equation}
\beta_i\Delta\frac{l(d_{xi})}{l(d_{ii})}>\delta_i.
\label{eq:change-condition}
\end{equation}
All of these variables may be modeled according to some probability distribution. 

Throughout this paper we assume that the path loss function $l(d)$ is a deterministic function of the distance $d$:
\begin{equation}
l(d) = \frac{1}{d^{\alpha}}
\end{equation}
where $\alpha$ is the path loss exponent.

\subsection{Assumptions and Limitations}
\label{sec:assumptions}
We mentioned that the minimum power allocation scheme is used because when a feasible solution exists, this scheme will converge. Thus we are able to separate the domino effect from the cases when ripples of change occur due to the network not converging. This places constraints on the parameter values used in the simulation results throughout this paper. For instance, the density of a network seems to be an important parameter for the domino effect. It is reasonable to assume that dense networks will have larger domino effect than sparse networks, as users are closer to each other. However, we have not been able to simulate this aspect properly as with denser networks we need smaller SINR targets for the network to converge. For density issues, theoretical analysis seems actually more tractable.

Two other aspects are related to the transmit power distribution of the users. First, in practice each user has some maximum transmit power constraint. This will limit some users' ability to optimally adjust power after an update and also limit the domino effect. However, it is difficult to analytically determine which users are operating near or at the maximum power and thus the analytical work in this paper is an upper bound on the domino effect. In the simulation results we enforce a maximum power constraint to obtain more realistic results.

Secondly, according to (\ref{eq:belong-a1})-(\ref{eq:change-condition}), whether or not a user is affected by a change is determined by a difference in old and new transmit power. In the subsequent sections it will become apparent why this is done. However, from an engineering perspective it would be more useful to have such a criteria in dB. This again depends on the transmit power distribution of the users. E.g. two users may be transmitting at 0.01 W and 0.1 W. Both might have a difference of 0.05 between old transmit power and new transmit power after a change. However, for the first one this corresponds to a change of 7.8 dB, for the other 0.2 dB change.

\begin{figure}[t!]
\centering
\includegraphics[width = 0.8\columnwidth]{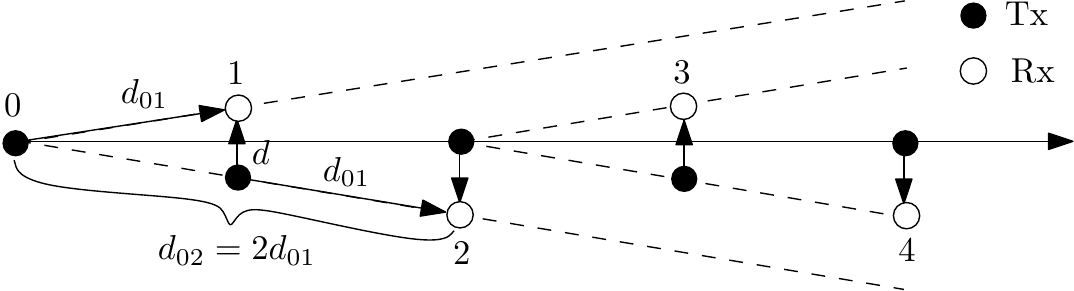}
\caption{Illustration of the array model. Distance between transmitter $i$ and receiver $i$, $d_{ii}$ is the same for all $i$ and $d_{ji} = d_{j+1,i+1} = d_{j-1,i-1}$.}
\label{fig:array-model-ipe}
\end{figure}

\section{Domino Effect in a Deterministic Model}
\label{sec:deterministic}
%\subsection{Array Model}
To illustrate the concept of the domino effect we start with a simple network model where transmitters and receivers are located along a line segment as illustrated in Fig. \ref{fig:array-model-ipe}. We assume the number of transmitter and receiver pairs are infinite. At a time $t_0$ transmitter $0$ updates its power by $1$ unit. We assume that at time $t\leq t_0$ the transmit power at each transmitter has converged such that changes in the network only occur due to the change at transmitter $0$ at time $t_0$.

We let $\beta_i = 1$ and $d_{ii} = d$ for all $i$ and set $\Delta = 1$. From Fig. \ref{fig:array-model-ipe} we have that 
\begin{align}
d_{i,j+1} &= 2d_{i,j} \\
d_{ij+2} &= \sqrt{9d^2_{i,j} -2d^2}. 
\end{align}
However, to simplify the expressions we assume $d_{i,j} = d_{i,j+1}/2 = d_{i,j+2}/3...$.% and let $g_{ii} = g = d^{-\alpha}$ be the same for all $i$..

\begin{figure}[t]
\centering
\includegraphics[width = 0.7\columnwidth]{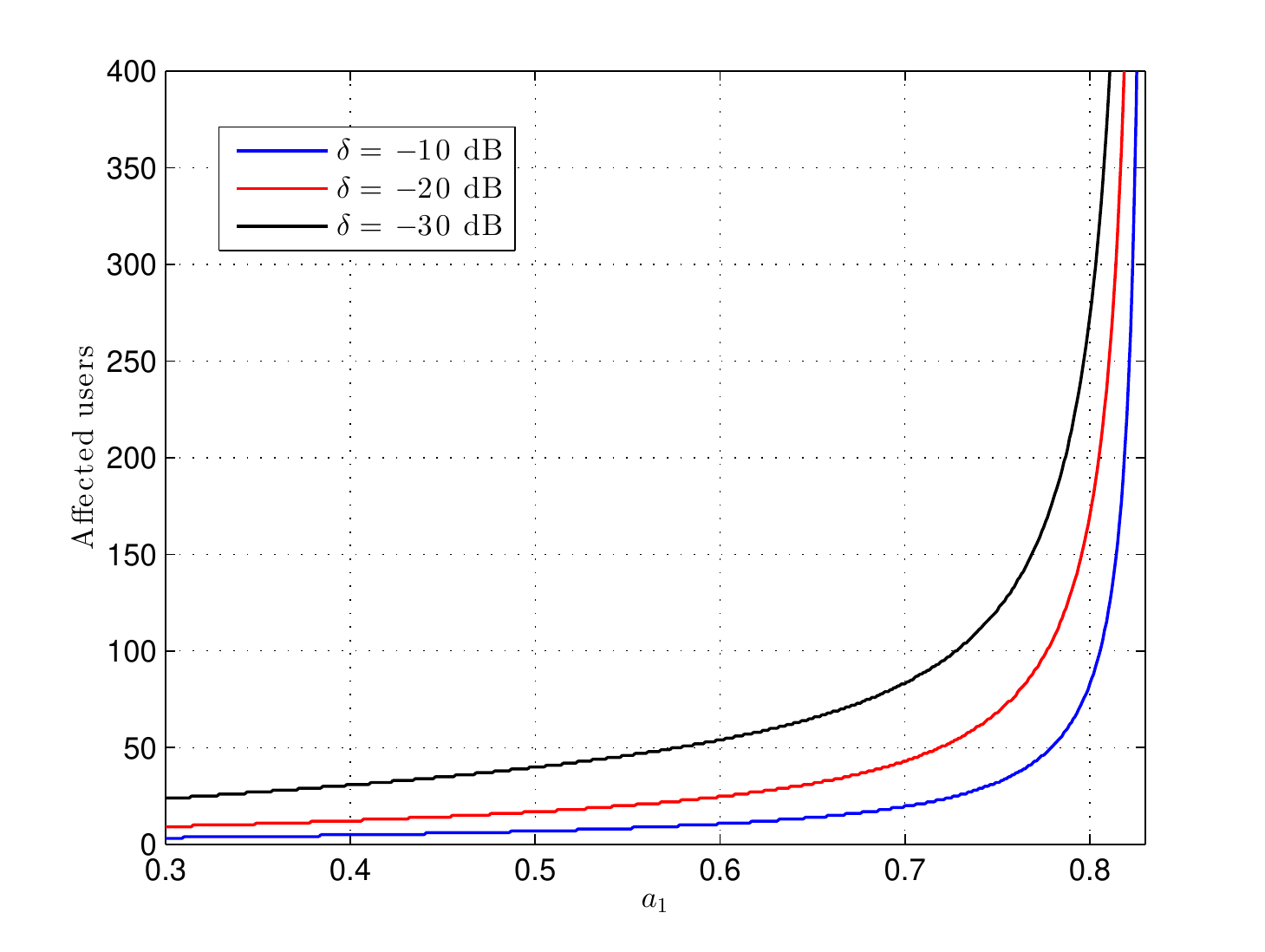}
\caption{Affected number of users as function of $a_1$ and three values of $\delta$ ($0.1 = -10$ dB, $0.01 = -20$ dB and $0.001 = -30$ dB). The path loss exponent $\alpha = 3$.}
\label{fig:array-model}
\end{figure}
%To simplify the model we assume that $d_{ij} = 2d_{i,j+1} = 3d_{i,j+2}...$ and let $g = g_{ii}$ be the same for all $i$. 

We now have that transmitter $1$ will increase its transmit power if $l(d_{01})/l(d)>\delta$. Transmitter $2$ will increase its transmit power if
\begin{equation*}
(l(d_{02})+l(d_{12})\frac{l(d_{01})}{l(d)})/l(d) = \frac{l(d_{01})}{4l(d)}+\frac{l(d_{01})^2}{l(d)^2}>\delta.
\end{equation*}
We see that condition for whether transmitter $n$ will increase its power becomes a sum  depending on $l(d_{01})$ and $l(d)$. To simplify notation, let $a_1 = l(d_{01})/l(d)$, which defines how closely the links are located. Assuming all transmitters up to $n-1$ have increased their power, transmitter $n$ will increase its power if
\begin{equation}
a_n = a_1\sum_{j=0}^{n-1}\frac{a_j}{(i-j)^{\alpha}}>\delta
\end{equation}
For a given $a_1$ and $\delta$ we want to know if there exists a $n$ such that $a_n < \delta$. This would mean that the change of transmitter $0$ affects a finite number of links, although this number may be large.

Fig. \ref{fig:array-model} shows the number of affected users as a function of $a_1$ and $\delta$ for a path loss exponent equal to 3. Decreasing $\delta$ leads more users to be affected by a change. We also see that for each $\delta$ there exists a divergence value for $a_1$ for which the number of affected users goes to infinity.

\section{Poisson Distributed Network}
\label{sec:percolation}
We now let the set of users be distributed over a 2-dimensional area according to a homogeneous Poisson Point Process (PPP) $\Phi$. Specifically we let the transmitter locations be distributed according to a homogeneous PPP with density $\lambda$. To avoid the frequency at which there is no feasible allocation, we let the receiver locations be dependent on the transmitter locations. Thus the receiver locations are also distributed according to a PPP, but these two distributions are not independent.

\begin{figure}[t]
\centering
\includegraphics[width = 0.7\columnwidth]{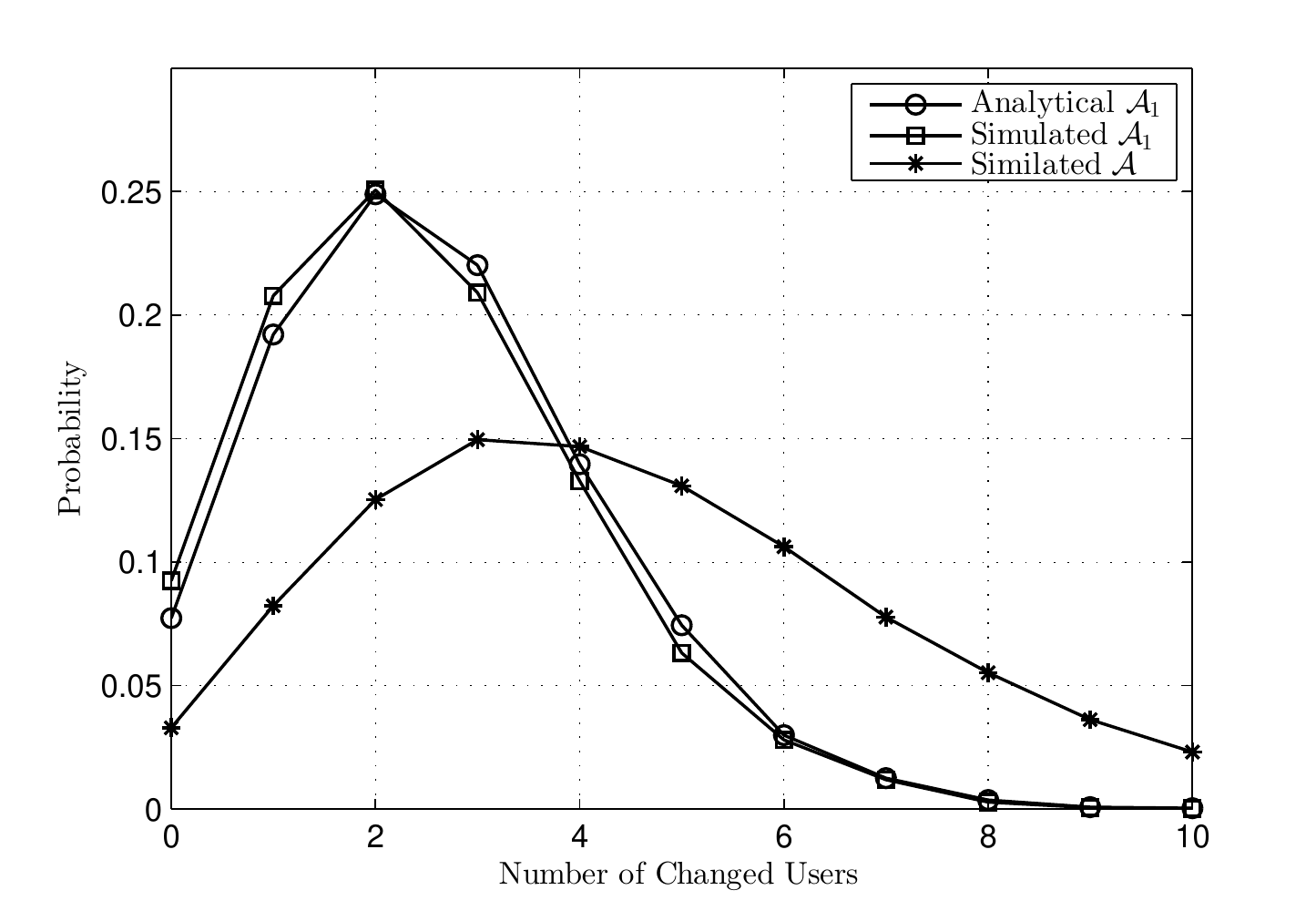}
\caption{Analytical and simulated distributions of $\mathcal{A}_1$ referenced against simulated $\mathcal{A}$ with $\lambda = 4\times 10^{-4}$, $\beta_i = 1$, $d_{ii} = 1/10$, $\alpha=3$ and $\delta_i = -20$ dB for all $i$ and $\Delta = 1$.}
\label{fig:deterministic}
\end{figure}

\subsection{Characterizing $\mathcal{A}_1$}
\label{sec:3-a}
To illustrate how changes in the network can affect different number of users we simplify the model by setting $\beta_i = 1$, $d_{ii} = 1/10$ and $\delta_i = 0.01$ for all $i$ and set $\Delta = 1$.

The probability of user $i$ belonging to $\mathcal{A}_1$ is thus the probability that the path loss between transmitter $x$ and receiver $i$ is larger than $10^{-5}$. By assuming $\alpha=3$, it is the probability that transmitter $x$ is closer than 46.4 meters to receiver $i$. %Since we assume that there is at least one point in this area (the node that does the update) we get the following distribution and mean value for $|\mathcal{A}_1|$.
A property of Slivnyak's theorem for PPPs is that the properties seen from a point $x \in R^2$ are the same whether we condition on having a point of the PPP at $x$ or not \cite{Haenggi2009}. We therefore have
\begin{align}
&|\mathcal{A}_1| =  \Phi(46.4^2\pi\lambda \backslash \{x\}) = \Phi(46.4^2\pi\lambda) \\
&\mathbb{E}[|\mathcal{A}_1|] = \mathbb{E}[\Phi(46.6^4\pi\lambda\backslash\{x\})] = \mathbb{E}[\Phi(46.4^2\pi\lambda)].
\end{align}

Fig. \ref{fig:deterministic} shows the analytical and simulated distribution of $|\mathcal{A}_1|$ for a density of nodes of $4\times 10^{4}$, referenced against the simulated distribution of $|\mathcal{A}|$. As expected $E[|\mathcal{A}|]>E[|\mathcal{A}_1|$ and the probability of inflicting a large change in the network is larger for $|\mathcal{A}|$ than for $|\mathcal{A}_1|$.

%\subsection{Network Load}
%In the previous example we assumed one user updated its transmit power by 1 unit and where able to model the number of initially affected users according to a Poisson distribution. However, this analysis only holds as long as affected user is able to adapt as expected. A situation where this is not the case is if a user already transmits at maximum transmit power. Thus, the extent to which a change affects the network depends on the network load, or in other words how saturated the network is.
\begin{figure}[t]
\centering
\includegraphics[width = 0.7\columnwidth]{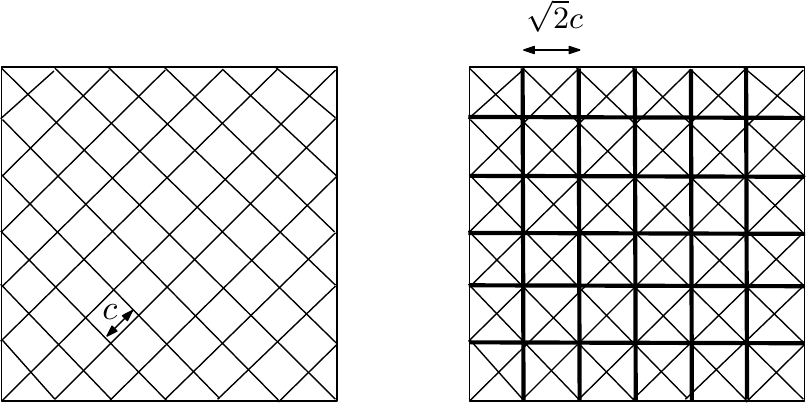}
\caption{The plane divided into squares on the left hand side. On the right hand side, each square is assigned to an edge (bold lines) of a bond percolation model.}
\label{fig:bond-percolation}
\end{figure}
\subsection{Continuum Percolation and Infinite Propagation of Change}
From the estimate of $|\mathcal{A}_1|$ in the previous section we saw that $\mathbb{E}[|\mathcal{A}_1|]$ is finite for all finite values of $\lambda$. As expected, $|\mathcal{A}|$ differed from $|\mathcal{A}_1|$ and we would like to characterize $|\mathcal{A}|$ in some way. Unfortunately we are not able to provide any results on the distribution or moments of $|\mathcal{A}|$, but we are able to provide results regarding the finiteness/infiniteness of $|\mathcal{A}|$. To do this we use standard techniques from continuum percolation \cite{Haenggi2009} to map the random Poisson network on a lattice and use results from bond percolation \cite{CambridgeJournals:2048852}. The result is as follows:% to prove that there exists a critical value for $\lambda$, $\lambda_c$, for which values of $\lambda<\lambda_c$ the propagation of change is almost surely (a.s.) finite, and for values of $\lambda>\lambda_c$ the propagation of change is infinite a.s.
\begin{theorem}
There exists a critical value for $\lambda$, $\lambda_c$, for which values of $\lambda<\lambda_c$ the propagation of change is almost surely (a.s.) finite, and for values of $\lambda>\lambda_c$ the propagation of change is infinite a.s.
\end{theorem}
\begin{proof}
Define $r$ as
\begin{equation}
r(\xi) \triangleq \max\{d:\beta\xi \frac{l(d)}{l(d_{ii})}>\delta\}
\end{equation}
for some variable $\xi$. We now say that two points $i$ and $j$ are connected if the distance $d_{ij}$ between them is less or equal to $r(\xi)$ for some $\xi$.
We prove the theorem in two parts: first the existence of infinite propagation of change for large $\lambda$ and secondly the absence of infinite propagation of change for small $\lambda$. %Thus, there must exist a critical value $\lambda_c$, such that for $\lambda<\lambda_c$ the propagation of change is finite a.s., while for $\lambda>\lambda_c$ the propagation of change is infinite a.s.
\subsubsection{Existence of infinite propagation of change for large $\lambda$}
Assume that the power update of all nodes except the initial node is $\delta$. If we can prove the existence of infinite propagation of change in this case, it must also hold for other update values as $\delta$ is the lowest possible value. We now map the model onto a bond percolation model. The plane is divided into squares of size $c = r(\delta)/2\sqrt{2}$ as shown in Fig. \ref{fig:bond-percolation}. Each square corresponds to a potential edge of the bond percolation lattice of $\mathbb{Z}^2$. The edge is added if at least one point of the Poisson process falls into this square. Each edge therefore exists with probability $p = 1-\text{exp}(-\lambda c^2)$, which is independent of all other edges. It is known that the bond percolation model of the infinite square lattice in $\mathbb{Z}^2$ percolates if the edges between nearest neighboring points of the lattice exist with probability $p\geq1/2$ \cite{CambridgeJournals:2048852}. Thus, if $\lambda\geq \log 2/c^2$, the edge percolation model contains an infinite cluster a.s.

Now, if two edges are adjacent, it means that at least two points of the Poisson process are located in squares that share at least one corner, and therefore that the maximum distance between the two points is $2\sqrt{2}c = 2\sqrt{2}r(\delta)/2\sqrt{2} = r(\delta)$ and thus they can affect each other in the original model. As an infinite set of connected edges in the bond percolation model (i.e. infinite cluster) corresponds to an infinite set of connected points in the original model, there exists a path of infinite propagation of change a.s. if $\lambda\geq \log 2/c^2$.

\subsubsection{Absence of infinite propagation of change for small $\lambda$}
%To prove the absence of infinite propagation of change for small $\lambda$ we use the same technique as we used to show the existence of an infinite propagation of change for large $\lambda$, except we assume that the power update off all nodes is $\Delta$.
Let $r_0$ be the distance $r_0 = r(\Delta)$. Similarly, let $r_n$ be the distance $r_n = r(\Delta_n)$, where $\Delta_n$ will be defined below. Consider an initial node $o$ that does an update $\Delta$. The nodes initially affected by this change belong to the set $\mathcal{A}_1$. The expected number of nodes in $\mathcal{A}_1$ is $\lambda\pi r_0^2$. Now, set $\lambda<1/(\pi r_0^2)$ such that the expected number of affected nodes by the initial update is less than one. We now assume that only one node is affected in each "round", i.e. $|\mathcal{A}_n| = 1$.\footnote{This does not affect the total number of affected users. As an example consider the case where there would be 2 affected users in one round, i.e. $|\mathcal{A}_n|=2$. These two are affected independently of each other and we can divide this set into two sets $\mathcal{A}^1_n$ and $\mathcal{A}^2_n$. Now, all nodes that would be affected in $\mathcal{A}_{n+1}$ will still be affected as the accumulated change is the same.}. 
The maximum distance $r_n$ from a node in $\mathcal{A}_n$ to a node in $\mathcal{A}_{n+1}$ is $r(\Delta_n)$ and we have that $\Delta_n$ is given as
\begin{equation}
\Delta_n =2^{\alpha}\bigl(\frac{\Delta}{(\sum_{j=0}^{n-1}r_j)^{\alpha}} + \sum_{i=1}^{n-1}\frac{\delta}{(\sum_{j=i}^{n-1}r_j)^{\alpha}} + \delta\bigr)
\label{eq:delta-n}
\end{equation}
The factor $2^{\alpha}$ comes from the fact that with a maximum radius $r$, the expected distance of a node within this area is $r/2$. The question is now whether $r_{\infty}$ is finite or not. If $r_{\infty}$ is infinite then we will not be able to avoid percolation even for sufficiently small $\lambda$, as long as $\lambda>0$. Since $r_n$ is just a product depending $\Delta_n$, $r_n$ is finite if $\Delta_n$ is finite.

The last term in (\ref{eq:delta-n}) is a constant so we can drop this term. Also, for sufficiently large $n$ the first term is zero. Hence, we focus on the sum $\sum_{i=1}^{n-1}\delta/(\sum_{j=i}^{n-1}r_j)^{\alpha}$ as $n\rightarrow \infty$. We also have that
\begin{align}
\Delta^{'} &= \frac{\delta}{(\sum_{j=1}^{\infty}r_j)^{\alpha}}  + \sum_{i=2}^{\infty}\frac{\delta}{(\sum_{j=i}^{\infty}r_j)^{\alpha}} \label{eq:delta-prime-1} \\
&\leq \frac{\delta}{(\sum_{j=1}^{\infty}r_j)^{\alpha}} + \sum_{i=2}^{\infty}\frac{\delta}{(\sum_{j=1}^{i}r_j)^{\alpha}} \label{eq:delta-prime-2}
\end{align}
when $r_i\leq r_{i+1}, \forall i>1$. Now, let
\begin{equation}
S = \sum_{i=2}^{\infty}\frac{\delta}{(\sum_{j=1}^i r_j)^{\alpha}} = \delta \sum_{i=2}^{\infty} \biggl(\frac{1}{a_i}\biggr)^{\alpha}.
\end{equation}
The series $S$ converges if $a_i\geq i, \forall i>1$ and $\alpha>1$, since $\sum_{i = N}^{\infty} 1/i^{\alpha}$ converges for any $N>0$ when $\alpha>1$. $a_i>i, \forall i>1$ holds whenever $r_2\geq 2$.

Thus, as long as $r_2\geq 2$, $r_n<\infty$ for all $n$ and we can choose $\lambda$ small enough so that the average number of affected nodes in each round is always less than 1. This process is equivalent to a Galton-Watson process and it is known that when the average number of children per individual in a Galton-Watson process is smaller than one, the process dies out with probability 1 \cite{harris:2002}.

\textit{Remark:} The inequality between (\ref{eq:delta-prime-1}) and (\ref{eq:delta-prime-2}) holds as long as $r_i\leq r_{i+1}, \forall i>1$. If this is not the case for some $i$, then we have a maximum value of $r_i$ for some $i<\infty$. Since $i$ must be less than infinity, this maximum must also be less than infinity and we can choose $\lambda$ small enough so that the average number of affected nodes from this distance is less than one. Thus the result holds.
\end{proof}
Although we have proved that a critical density exists, we have not provided any value for it. From the proofs of both existence and absence of percolation we also see that such a critical density depends on the other system parameters $\beta_i$, $d_{ii}$ and $\delta$. For fixed $\beta_i$ and $d_{ii}$, the critical density increases with $\delta$.

%\begin{table}
%\centering
%\caption{Simulation parameter values}
%\begin{tabular}{|c|c|}
%\hline
%Network density ($\lambda$) & $4\times 10^{-4}$ \\
%SINR requirement ($\beta_i$) & 1 \\
%Desired transmitter receiver distance ($d_{ii}$) & 10 \\
%Update value ($\Delta$) & 1 \\
%Path loss exponent ($\alpha$) & 3 \\
%White noise power ($N_0$) & $10^{-8}$ \\
%Max transmit power W & 1 \\
%\hline
%\end{tabular}
%\label{tab:sim-parameters}
%\end{table}

\subsection{Simulation Results and Self-organized Criticality}
As we have not been able to characterize $|\mathcal{A}|$ analytically, we investigate the distribution of $|\mathcal{A}|$ through simulations in this section. However, we do have some expectations of the distribution of $|\mathcal{A}|$ depending on the distributions of the individual $|\mathcal{A}_i|$s. If $|\mathcal{A}_i|<<|\mathcal{A}_1|$, for $i>1$, we expect $|\mathcal{A}|$ to be almost Poisson distributed since $|\mathcal{A}_1|$ is Poisson distributed. If $\mathbb{E}[|\mathcal{A}_i|]>0$ it is less clear how $|\mathcal{A}|$ should be distributed. It is well known that the sum of independent Poisson variables are still Poisson. But in this case $|\mathcal{A}_2|$ is not independent of $|\mathcal{A}_1|$, since if a realization of $|\mathcal{A}_1|$ is small, we also expect $|\mathcal{A}_2|$ to be small. Whether or not the sum is Poisson depends on the conditional probability function \cite{jacod1975}. Through the simulation results, it seems like the resultant distribution is still Poisson.

\begin{figure}[t]
\centering
\includegraphics[width = 0.725\columnwidth]{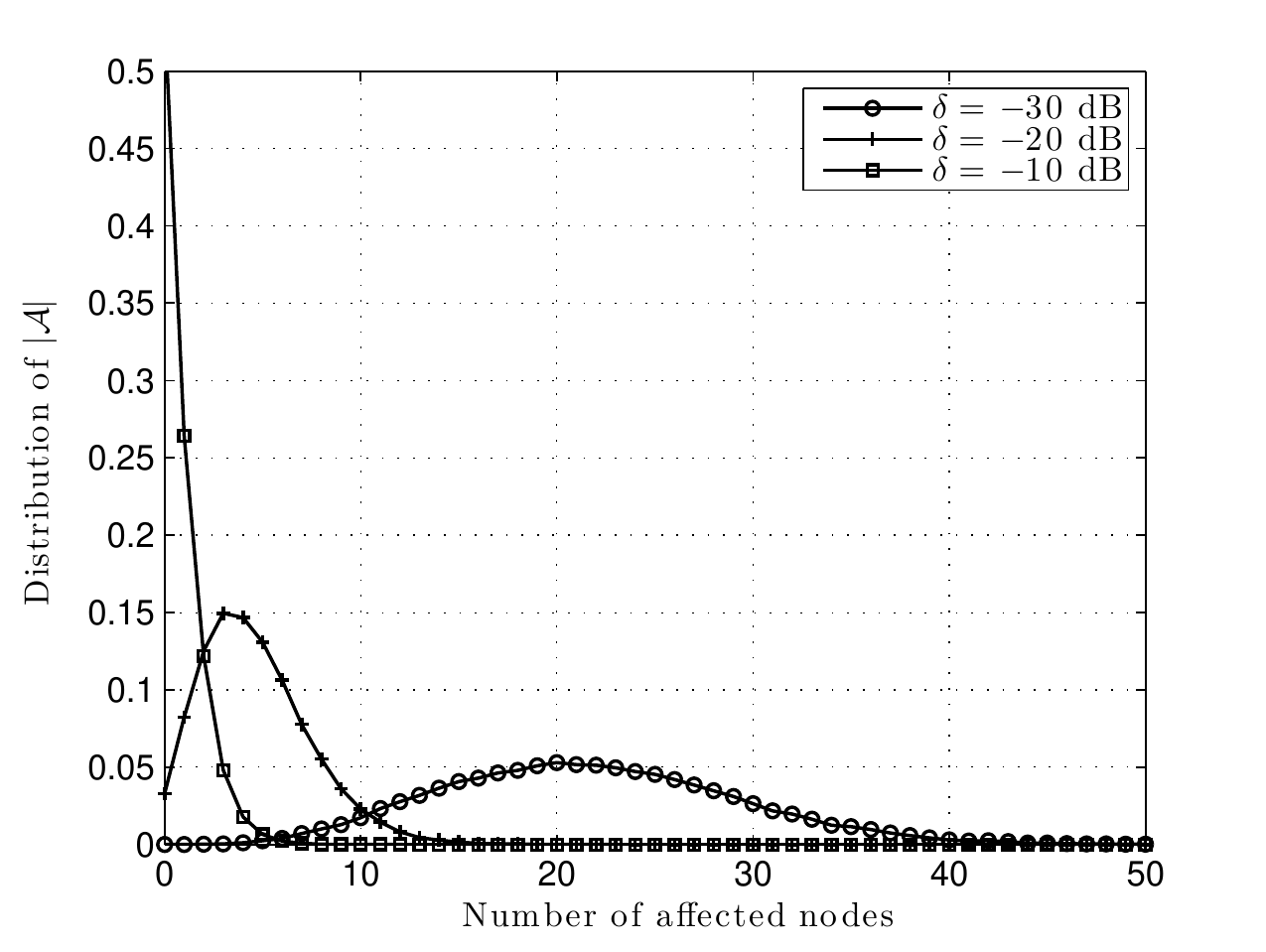}
\caption{Distribution of $|\mathcal{A}|$ for different values of $\delta$ and $\lambda = 4\times 10^{-4}$.}
\label{fig:different-As}
\end{figure}

\begin{figure}[t]
\centering
\includegraphics[width = 0.725\columnwidth]{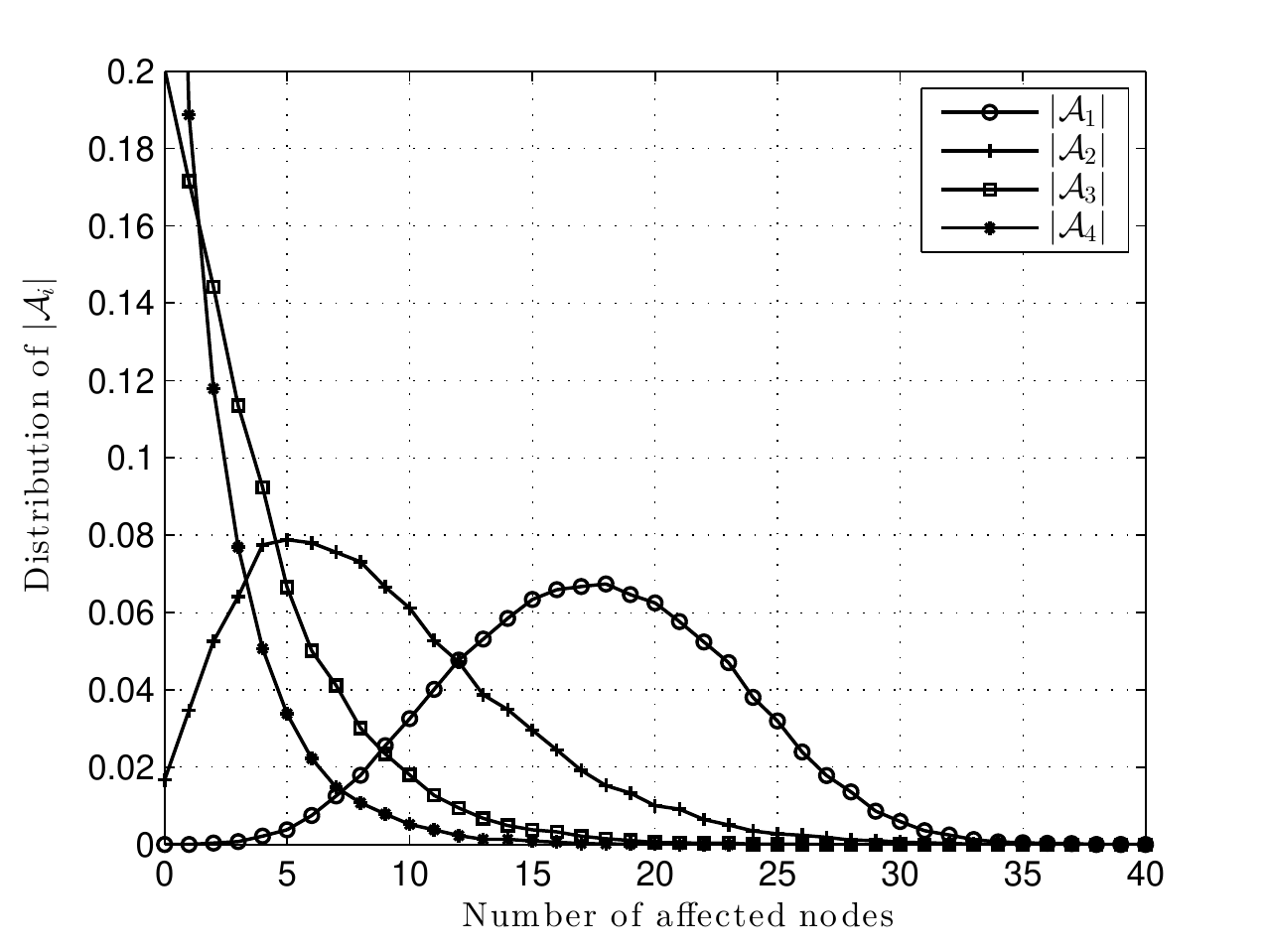}
\caption{Distribution of the different $|\mathcal{A}_i|$s for $\delta = -30$ dB and $\lambda = 4\times 10^{-4}$.}
\label{fig:different-Ais}
\end{figure}

\begin{figure}[t]
\centering
\includegraphics[width = 0.725\columnwidth]{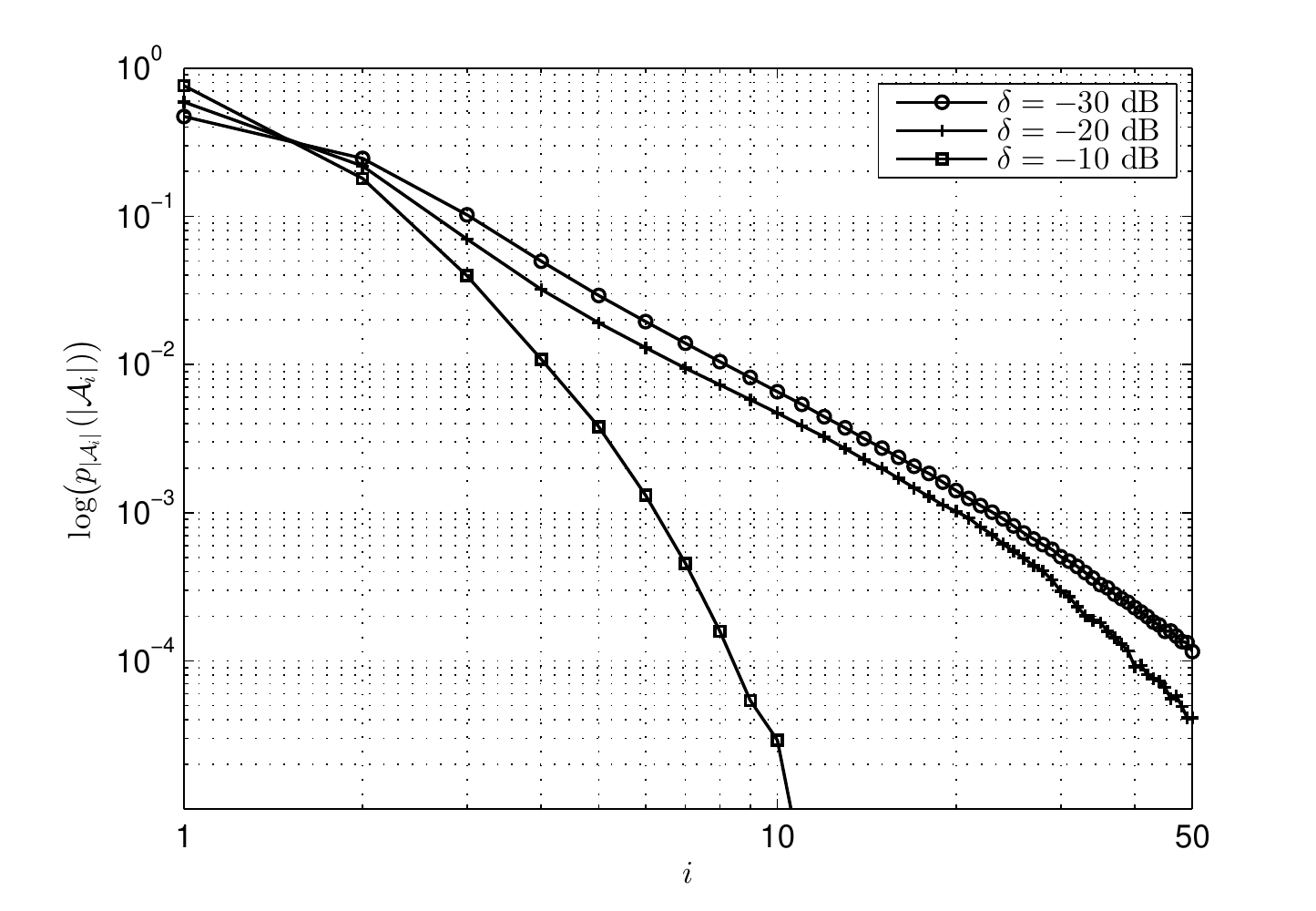}
\caption{Fraction of $|\mathcal{A}|$ in the different $|\mathcal{A}_i|$s in a loglog plot for different values of $\delta$ and $\lambda = 4\times 10^{-4}$.}
\label{fig:loglog}
\end{figure}

We have used the same parameters as in Section \ref{sec:3-a}, namely $\lambda = 4\times 10^{-4}$, $\beta_i = 1$, $d_{ii} = 10$, $\Delta = 1$, $\alpha = 3$, $N_0 = 10^{-8}$. %which is summarized in Table \ref{tab:sim-parameters}. 
As mentioned in Section \ref{sec:assumptions}, we did not assume a maximum power constraint in the analytical work in the previous sections. In the simulations we assume a maximum transmit power of 1 W, and thus even though a transmitter should update its power according to (\ref{eq:change-condition}), if the maximum power will be surpassed by such a change the transmitter will not update its power.

In Fig. \ref{fig:different-As} the distribution of $|\mathcal{A}|$ is plotted for different values of $\delta$. For high values of $\delta$, $|\mathcal{A}_i|<<|\mathcal{A}_1|$, for $i>1$ and the plots for $\delta = -10$ dB and $-20$ dB are almost Poisson distributed. %For a small $\delta = -30$ dB the distribution is better approximated by a Gaussian distribution.

In Fig. \ref{fig:different-Ais} the distribution of the different $|\mathcal{A}_i|$s for $\delta = -30$ dB and $\lambda = 4\times 10^{-4}$ is given. We see that all the $|\mathcal{A}_i|$s seem to follow a Poisson distribution with different expected values.

In Fig. \ref{fig:loglog} we have plotted which $\mathcal{A}_i$ a fraction of $\mathcal{A}$ belongs to. We see that this distribution follows a power law for small $\delta$. In relation to self-organized criticality, we see that the size of avalanches, where avalanche size is given as the maximum $i$ for which $|\mathcal{A}_i|$ is nonzero, follows the same characteristics of a system in a critical state. However, for $\delta = -10$ dB, the curve does not follow a power law distribution and we expect this state to be non-critical.

%\begin{figure}
%\centering
%\includegraphics[width = 0.7\columnwidth]{Fig/sat_vs_delta.pdf}
%\caption{Number of users achieving their SINR requirements for different values of $\delta$ as fraction of satisfied users with $\delta = 0.001$.}
%\label{fig:sat-vs-delta}
%\end{figure}
%Fig. \ref{fig:sat-vs-delta} shows the number of users achieving their SINR requirements for different values of $\delta$ as fraction of satisfied users with $\delta = 0.001$. From the previous results it seemed clear that to avoid having a large domino effect in the network, one should choose $\delta$ large. However, Fig. \ref{fig:sat-vs-delta} points out the downside of this approach, as increasing $\delta$ decreases the performance of the network. The result in Fig. \ref{fig:sat-vs-delta} is plotted for a given system configuration with $\beta_i = 5$ and $d_{ii} = 10$ $\forall i$, $\alpha = 3$, maximum transmit power of $1$ W for each transmitter, $N_0 = 10^{-8}$ and a density of $\lambda = 10^{-3}$. How performance is affected by $\delta$ depends on all these parameters and thus setting a sufficient margin to avoid a large domino effect must be done based on these parameters.

\begin{figure}[t]
\centering
\includegraphics[width = 0.725\columnwidth]{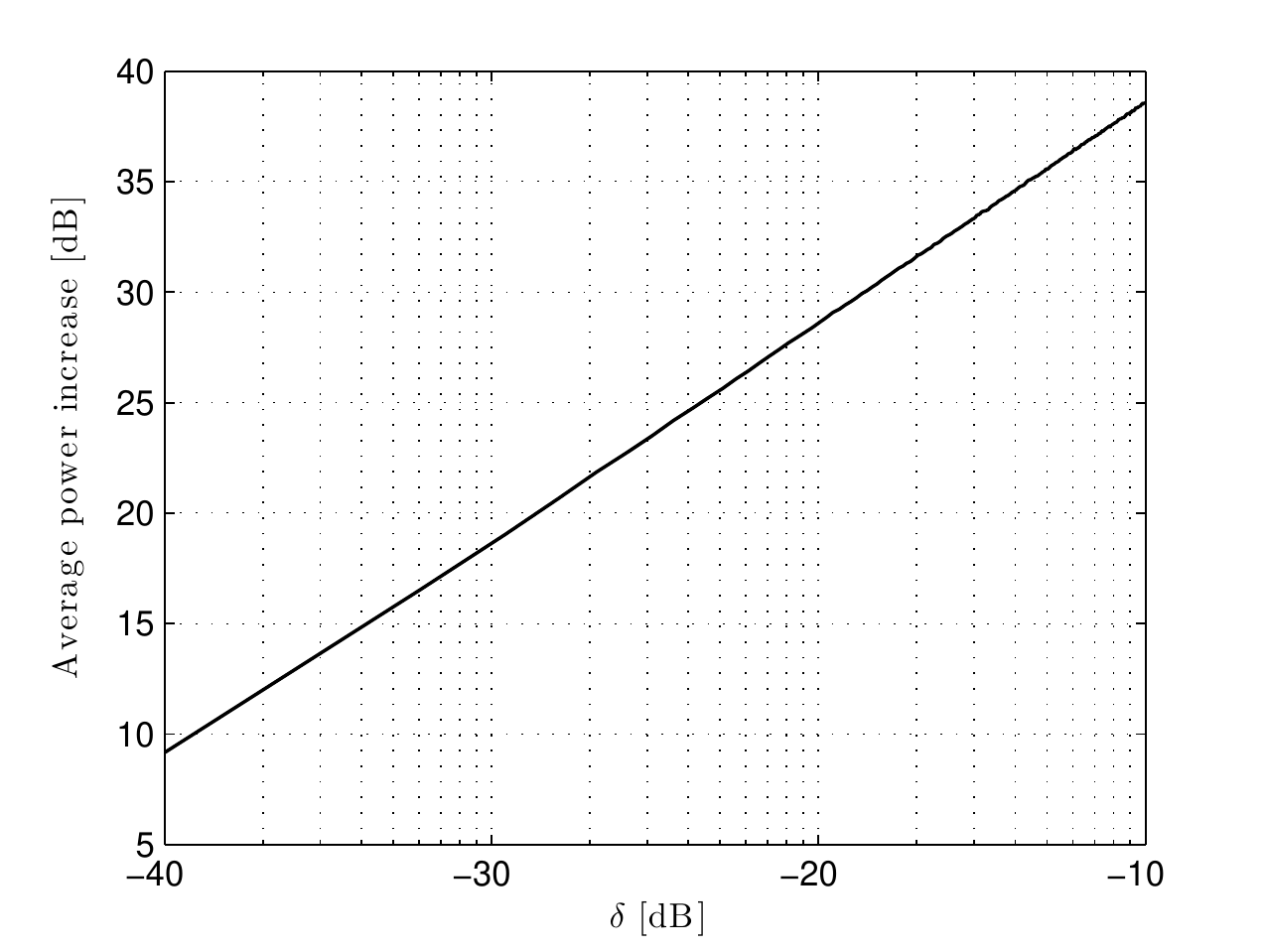}
\caption{Average power increase in dB for different values of $\delta$ with $\lambda = 4\times 10^{-4}$.}
\label{fig:avg-power-increase}
\end{figure}

From the results presented above, we see that a large number of users can be affected by a change in a manner similar to a domino effect. The number of users that are affected depends on the threshold $\delta$. For analytical purposes it was convenient to define the condition for when a user has to update its power (\ref{eq:change-condition}) on a difference between old power and new necessary power. However, from an engineering aspect it would be better to have this in terms of dB. In Fig. \ref{fig:avg-power-increase}, the average power increase in dB is shown for different values of $\delta$ for the same simulation parameters used in results Fig. \ref{fig:different-As}-\ref{fig:loglog}. As we see, even a modest threshold of $\delta = -30$ dB, leads to an average power increase over all users of 18 dB. This point out the difficulty in choosing a suitable threshold, as a high threshold leads to a low domino effect, but leads to an increased power which degrades the performance.

%\section{Self-organized Criticality with Power and Frequency Allocation}
%\label{sec:soc}

\section{Conclusion and Future Work}
\label{sec:conclusion}
In this paper we have investigated the issue of propagation of change in decentralized wireless networks, an issue largely lacking in the literature. This manifests itself as a domino effect in the network can arise due to new users entering the network or old ones leaving or updating their resource allocation. As most proposed resource allocation algorithms are dynamic and adaptive, they change according to the current state of the network in order to utilize the resources in an optimal manner. However this leads to ripples of change propagating through the network. In this paper we showed that there exists a critical network density, for which a larger density will propagate change infinitely almost surely, while smaller densities will propagate change finitely almost surely. Through simulation results we showed how the number of affected nodes follows different Poisson distributions in a network distributed according to a Poisson Point process, that depends on the threshold for when a node is considered to be affected. While the domino effect is a decreasing function of threshold, one cannot simply set a high threshold as this leads to significant increase in power and degradation of system performance. Thus the optimal threshold is a trade-off between having high system performance while still maintaining resilience to the domino effect. 

While this paper has highlighted some of the issues regarding change in decentralized wireless network, there are still important open issues. One is to give upper and lower bounds on the critical density for percolation. Another is to give a characterization of the trade-off one gets from limiting the domino effect while maximizing system performance. 

Lastly, this paper only considered power control in a Poisson distributed network. We did this because a simple power control algorithm is guaranteed to converge when a feasible solution exists and hence we were able to separate change due to convergence and change due to the domino effect. An equally important aspect is how frequency allocation affects the domino effect. As frequency is a discrete variable it might in some cases decrease the propagation of change as it can be seen as a thinning of the Poisson process. On the other hand, a slight change in power at one transmitter might result in a frequency change at another which can escalate the domino effect.

\bibliographystyle{ieeetr}
\bibliography{bib,library}{}

\end{document}